\documentclass[a4paper]{article}

\usepackage{fixltx2e}
\usepackage{isomath}

\usepackage[latin1]{inputenc}
\usepackage[T1]{fontenc}
\usepackage[british]{babel}

\makeatletter
\useshorthands{"}%
\defineshorthand{"-}{\nobreak-\bbl@allowhyphens}
\defineshorthand{"=}{\nobreak--\bbl@allowhyphens}
\makeatother

\usepackage{amsmath}
\usepackage[numbers]{natbib}
\usepackage[amsmath,standard]{ntheorem}
\usepackage{tikz}
\usetikzlibrary{calc}
\usetikzlibrary{chains}
\usepackage{url}

\usepackage{microtype}

\usepackage{hyperref}

\newcommand*{\LEP}[1]{#1_L}
\newcommand*{\PROBLEM}[1]{\textsc{#1}}
\newcommand*{\REP}[1]{#1_R}
\newcommand*{\SET}[1]{\{#1\}}
\newcommand*{\SETC}[2]{\{#1\mid#2\}}
\newcommand*{\VCENTER}[1]{\raisebox{-.5\height}{#1}}

\makeatletter
\newcommand*{\CIRCLE}[3][]{
  \newcounter{length}
  \setcounter{length}{0}
  \foreach \x in {#3}{%
    \addtocounter{length}{1}%
  }
  \pgfmathsetmacro{\angle}{360/\the\value{length}};
  \foreach[count=\i] \x in {#3}{%
    \pgfmathsetmacro{\lox}{\i-1}
    \pgfmathsetmacro{\alpha}{90+(\i-1)*\angle}
    \coordinate (c\i) at (\alpha:#2);
    \node[anchor=center] at (\alpha:#2+10) {\ifx\\#1\\\x\else$#1_{\x}$\fi};
  }
  \draw (0,0) circle (#2);
}
\makeatother

\makeatletter
\newenvironment{DECISIONPROBLEM}[1]{%
  \newcommand*{\INSTANCE}{\item\relax\textit{Instance:}\enskip}
  \newcommand*{\QUESTION}{\item\relax\textit{Question:}\enskip}
  \begin{trivlist}\item\relax\PROBLEM{#1}\@afterheading}{\end{trivlist}}
\makeatother

\title{Maximum Pagenumber-$k$ Subgraph\\is NP-Complete}

\begin{document}

\author{Peter Jonsson\\Department of Computer and Information Science\\Linköping University \and Marco Kuhlmann\\Department of Computer and Information Science\\Linköping University}

\maketitle

\begin{abstract}
  Given a graph~$G$ with a total order defined on its vertices, the
  \PROBLEM{Maximum Pagenumber-$k$ Subgraph Problem} asks for a maximum
  subgraph $G'$ of $G$ such that $G'$ can be embedded into a $k$"-book
  when the vertices are placed on the spine according to the specified
  total order.
  We show that this problem is NP-complete for $k \geq 2$.
\end{abstract}

\section{Introduction}

A \emph{k-book} is a collection of $k$ half-planes, all of which have
the same line as their boundary.
The half"-planes are called the \emph{pages} of the book and the
common line is called the \emph{spine}.
A \emph{book embedding} is an embedding of a graph into a $k$-book
such that the vertices are placed on the spine, every edge is drawn on
a single page, and no two edges cross each other.
The \emph{pagenumber} of a graph~$G$ is the smallest number of pages
for which~$G$ has a book embedding.

Computing the pagenumber of a graph is an NP-complete problem
\cite{DBLP:conf/stacs/Unger88}, and it is even NP-complete to verify
if a graph has a certain pagenumber~$k$, for fixed $k \geq 2$.
Verifying that a graph has pagenumber~$1$ however can done in
polynomial time: A graph~$G$ has pagenumber~$1$ if and only if it is
outerplanar\footnote{An undirected graph is \emph{outerplanar} if and
  only if it has a crossing"-free embedding in the plane such that all
  vertices are on the same face.} \citep{Chung:etal:sijadm87}, and
outerplanarity can be checked in linear time
\citep{DBLP:journals/ipl/Mitchell79}.

In certain applications, the order of the vertices along the spines is
not arbitrary but specified in the input.
In this case we can still check whether a graph has a $1$-book
embedding that respects~${\prec}$ in linear time.
Assume we have a graph $G = (\SET{v_1, \dots, v_m}, E)$ and spine
order $v_1 \prec \dots \prec v_m$. Extend~$E$ with the edges
$\SETC{\SET{v_i, v_{i+1}}}{1 \leq i \leq m-1} \cup \SET{\SET{v_m,
    v_1}}$
and note that $G = (V, E)$ can be embedded (in a way that respects
$\prec$) into a $1$-book if and only if the extended graph is
outerplanar.
It can also be checked in linear time whether a graph has
pagenumber~$2$ \citep{Haslinger:Stadler:bmb99}.

We are interested in the complexity of the following problem:

\begin{DECISIONPROBLEM}{(Fixed-Order) Maximum Pagenumber-$k$ Subgraph}
  \INSTANCE An undirected graph $G = (V, E)$, a total
  ordering~${\prec}$ on~$V$, and an integer $m \geq 0$.

  \QUESTION Is there a subset $E' \subseteq E$ such that $|E'| \geq m$
  and $G' = (V, E')$ can be embedded into a $k$-book such that the
  vertices in~$V$ are placed on the spine according to the total
  order~${\prec}$?
\end{DECISIONPROBLEM}

\noindent For $k = 1$, this problem can be solved in time $O(|V|^3)$
using dynamic programing \citep{DBLP:journals/corr/abs-1504-04993}.
Here we show that, for $k \geq 2$, the problem is NP"-complete, and
remains so even if we restrict solutions to acyclic subgraphs.
That is, the following problem is NP"-complete for $k \geq 2$:

\begin{DECISIONPROBLEM}{Maximum Acyclic Pagenumber-$k$ Subgraph}
  \INSTANCE A directed graph $G = (V, A)$, a total ordering~${\prec}$
  on~$V$, and an integer $m \geq 0$.

  \QUESTION Is there a subset $A' \subseteq A$ such that
  \begin{enumerate}
  \item $|A'| \geq m$,
  \item $(V, A')$ is acyclic, and
  \item $(V, A')$ can be embedded into a $k$-book such that the
    vertices in~$V$ are placed on the spine according to the total
    order~${\prec}$?
  \end{enumerate}
\end{DECISIONPROBLEM}

\section{Circle graphs}

The following is largely based on \citet{DBLP:conf/stacs/Unger92}.
A \emph{circle graph} is the intersection graph of a set of chords of
a circle.
That is, its vertices can be put in one"-to"-one correspondence with a
set of chords in such a way that two vertices are adjacent if and only
if the corresponding chords cross each other.
An example is shown in Figure~\ref{fig:CircleGraphs}.

\begin{figure}
  \centering\small
  \VCENTER{%
    \begin{tikzpicture}
      \def\unit{.5ex};
      \def\shiftunit{.75ex};
      \coordinate (c1) at (0,0);
      \filldraw (c1) circle (\unit);
      \node[anchor=north,yshift=-\shiftunit] at (c1) {$c_1$};
      \coordinate (c2) at (1,0);
      \filldraw (c2) circle (\unit);
      \node[anchor=north,yshift=-\shiftunit] at (c2) {$c_2$};
      \coordinate (c3) at (2,0);
      \filldraw (c3) circle (\unit);
      \node[anchor=north,yshift=-\shiftunit] at (c3) {$c_3$};
      \coordinate (c4) at (1,1);
      \filldraw (c4) circle (\unit);
      \node[anchor=south,yshift=\shiftunit] at (c4) {$c_4$};
      \coordinate (c5) at (2,1);
      \filldraw (c5) circle (\unit);
      \node[anchor=south,yshift=\shiftunit] at (c5) {$c_5$};
      \draw (c1) -- (c2);
      \draw (c2) -- (c3);
      \draw (c1) -- (c4);
      \draw (c2) -- (c4);
      \draw (c2) -- (c5);
      \draw (c4) -- (c5);
    \end{tikzpicture}
  }
  \quad
  \VCENTER{%
    \begin{tikzpicture}
      \CIRCLE[c]{0.8cm}{5L,1L,2L,4L,1R,5R,3L,2R,3R,4R}
      \draw (c1) -- (c6);
      \draw (c2) -- (c5);
      \draw (c3) -- (c8);
      \draw (c4) -- (c10);
      \draw (c7) -- (c9);
    \end{tikzpicture}
  }
  \quad
  \VCENTER{%
    \begin{tikzpicture}[start chain,node distance=-2mm,arc/.style={out=90,in=90}]
      \foreach \x in {5L,1L,2L,4L,1R,5R,3L,2R,3R,4R} {
        \node[on chain,name=c\x] {$c_{\x}$};
      }
      \foreach \x in {1,...,5} {
        \draw (c\x L.north) edge[arc] (c\x R.north);
      }
    \end{tikzpicture}
  }
  \caption{A circle graph (left) together with a corresponding chord
    drawing (middle) and overlap model (right). The endpoints of a
    chord $c_i$ are named $c_{iL}$ (left) and $c_{iR}$ (right) based
    on a counter"-clockwise linear ordering of the vertices around the
    circle (chord drawing) or from left to right (overlap model).}
  \label{fig:CircleGraphs}
\end{figure}
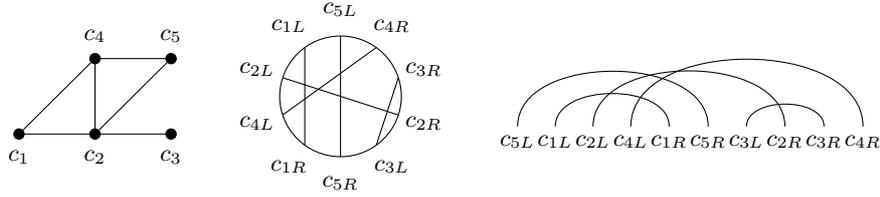

Circle graphs are often represented in a different way.
We assume henceforth, without loss of generality, that no two chords
have a common endpoint on the circle.
To obtain a so"-called \emph{overlap model} for a circle graph, we
start with a chord drawing and do the following:
\begin{enumerate}
\item we break up the circle and straighten it out into a line and
\item we turn the chords into arcs above the line that represents the
  circle.
\end{enumerate}

It is easy to see that the standard representation and the overlap
representation are equivalent.
For simplicity, we will still call the arcs in the overlap
representation chords.
Given a chord~$x$ in an overlap representation, we will denote its
left endpoint with $\LEP{x}$ and its right endpoint with $\REP{x}$. We
assume without loss of generality that the endpoints are represented
by positive integers.
Formally, we can now define circle graphs in the overlap
representation as follows.
An undirected graph $G = (V, E)$ with $V = \SET{v_1, \dots, v_n}$ is a
circle graph if and only if there exists a set of chords
\begin{displaymath}
  C = \SETC{(\LEP{c}, \REP{c})}{\text{$1 \leq i \leq n$ and $\LEP{c} < \REP{c}$}}
\end{displaymath}
such that
\begin{displaymath}
  \SET{v_i, v_j} \in E
  \quad\text{if and only if}\quad
  c_i \otimes c_j
\end{displaymath}
where the Boolean predicate ${\otimes}$ denotes the intersection of
chords, i.e.\ $c \otimes d$ if $\LEP{c} < \LEP{d} < \REP{c} < \REP{d}$
or $\LEP{d} < \LEP{c} < \REP{d} < \REP{c}$.
Given a circle graph $G = (V, E)$ with an overlap representation~$C$,
we let $C(v)$ (where $v \in V$) denote the chord corresponding to~$v$.

\section{Proof}

Given a graph $G = (V, E)$ and a subset $V' \subseteq V$, we let
$G | V'$ denote the subgraph of~$G$ induced by~$V'$, i.e.\ $G | V$ has
vertex set~$V'$ and the edge set contains exactly those edges in~$E'$
that have both their endpoints in~$V'$.

Our starting point is the following problem.

\begin{DECISIONPROBLEM}{$k$-Colourable Induced Subgraph Problem for
    Circle Graphs ($k$-CIG)}
  \INSTANCE A circle graph $G = (V, E)$ and an integer $m \geq 0$.

  \QUESTION Is there a subset $V' \subseteq V$ such that $|V'| \geq m$
  and the graph $G|V'$ is $k$-colourable?
\end{DECISIONPROBLEM}

\citet{DBLP:journals/tcad/CongL91} have shown that $k$-CIG is
NP-complete when $k \geq 2$.
In the cases when $k = 2$ and $k \geq 4$, this result is based on
earlier results by \citet{DBLP:journals/tcad/SarrafzadehL89} and
\citet{DBLP:conf/stacs/Unger88}, respectively.

We now show that \PROBLEM{Maximum Pagenumber-$k$ Subgraph} is
NP-complete by a reduction from $k$-CIG.

\begin{proposition}
  \PROBLEM{Maximum Pagenumber-$k$ Subgraph} is NP-complete.
\end{proposition}

\begin{proof}
  We first show that \PROBLEM{Maximum Pagenumber-$k$ Subgraph} is in
  NP.
  Given an instance $((V, E), {\prec},m)$ of this problem,
  non"-deterministically guess a subset $E' \subseteq E$ and a
  partitioning $E'_1, \dots, E'_m$ of~$E'$.
  The instance has a solution if and only if
  $(V, E'_1), \dots, (V, E'_k)$ have pagenumber~$1$ under the spine
  order~${\prec}$.
  This property can be checked in polynomial time as was pointed out
  earlier.

  We next prove NP-hardness via a polynomial"-time reduction from
  $k$-CIG. 
  Arbitrarily choose a circle graph $G = (V, E)$ and an integer
  $m \geq 0$.
  Construct (in polynomial time) an overlap model~$C$ of $G$ using the
  algorithm of \citet{DBLP:journals/jal/Spinrad94}.
  This overlap model defines a graph $H = (W, F)$ as follows:
  \begin{align*}
    W &= \SETC{\LEP{v}, \REP{v}}{\text{$v \in V$ and $C(v) = (\LEP{v}, \REP{v})$}}\\
    F &= \SETC{\SET{\LEP{v}, \REP{v}}}{\text{$v \in V$ and $C(v) = (\LEP{v}, \REP{v})$}}
  \end{align*}
  Finally, let~${\prec}$ be the natural linear ordering on~$W$.
  We claim that $(H, \prec, m)$ has a solution if and only if $(G, m)$
  has a solution.

  Assume that $(H, {\prec}, m)$ has a solution set of edges~$X$.
  Assume that $X = X_1 \cup \dots \cup X_k$ where the edges in~$X_1$
  are assigned to page~1, the edges in~$X_2$ to page~2, and so on.
  Construct a solution set of vertices $Y$ for $(G, m)$ as follows:
  $Y = \SETC{v \in V}{\SET{\LEP{v},\REP{v}} \in X}$.
  Obviously, $|Y| \geq m$.
  We show that $G|Y$ is $k$-colourable.
  Colour the edges in~$X_i$, $1 \leq i \leq k$ with colour $i$.
  Each edge corresponds to a vertex in~$Y$; let this vertex inherit
  its colour.
  Consider an arbitrary edge $(v, w)$ appearing in $G | Y$.
  Since $(v, w) \in E$, the chords $C(v) = (\LEP{v}, \REP{v})$ and
  $C(w) = (\LEP{w}, \REP{w})$ intersect.
  Thus, the edges $\SET{\LEP{v}, \REP{v}}$ and
  $\SET{\LEP{w}, \REP{w}}$ in~$F$ cannot be placed on the same book
  page which implies that~$v$ and~$w$ are assigned different colours.

  Assume that $(G, m)$ has a solution set of vertices~$V'$.
  Let $f\mathpunct{:}\ V' \to \SET{1,\dots,k}$ be a $k$-colouring of
  $G|V'$.
  Construct a solution set of edges $T$ for $(H, {\prec}, m)$ as
  follows:
  \begin{displaymath}
    T = \SETC{\SET{\LEP{v}, \REP{v}}}{\text{$v \in V'$ and
        $(\LEP{v}, \REP{v}) = C(v)$}}
  \end{displaymath}
  Obviously, $|T| \geq m$.
  We show that $(W, T)$ can be embedded into a $k$-book with spine
  order~${\prec}$.
  Pick an arbitrary edge $e = \SET{\LEP{v}, \REP{v}} \in T$ and put
  $e$ on page $f(v)$.
  Assume now that edges $e = \SET{\LEP{v}, \REP{v}}$ and
  $e' = \SET{\LEP{v'}, \REP{v'}}$ in~$T$ cross each other, i.e.\ that
  $(\LEP{v}, \REP{v}) \otimes (\LEP{v'}, \REP{v'})$ and $e, e'$ appear
  on the same page.
  This implies that $f(v) = f(v')$ and that there is an edge
  between~$v$ and~$v'$ in~$G$. Since $v, v' \in V'$, we see that this
  edge appears in $G|V'$. This contradicts the fact that~$f$ is a
  $k$-colouring of $G|V'$.
\end{proof}

\begin{corollary}
  \PROBLEM{Maximum Acyclic Pagenumber-$k$ Subgraph} is NP-hard when
  $k \geq 2$.
\end{corollary}

\begin{proof}
  Polynomial-time reduction from \PROBLEM{Maximum Pagenumber-$k$
    Subgraph}.
  Let $((V, E), {\prec}, m)$ be an arbitrary instance of
  \PROBLEM{Maximum Pagenumber-$k$ Subgraph}.
  We first show that \PROBLEM{Maximum Acyclic Pagenumber-$k$ Subgraph}
  is in~NP.
  Non-deterministically guess a subset $A' \subseteq A$ and a
  partitioning $A'_1, \dots, A'_m$ of~$A'$.
  Let $E'_i$ denote the corresponding set of undirected edges, i.e.\
  $E'_i = \SETC{\SET{v,w}}{(v,w) \in A'_i}$, and note that the
  instance has a solution if and only if $(V, E'_1), \dots, (V, E'_k)$
  have pagenumber~1 under the spine order~${\prec}$.
  
  We continue by proving NP-hardness.
  Assume without loss of generality that $V = \SET{1, \dots, m}$.
  Construct a directed graph $(V, A)$ as follows: the arc $(i, j)$ is
  in~$A$ if and only if $i < j$ and the edge $\SET{i, j}$ is in~$E$.
  Note that $(V, A)$ (and consequently every subgraph) is acyclic.
  It is now easy to verify that $((V, A), {\prec}, m)$ has a solution
  if and only if $((V, E), {\prec}, m)$ has a solution.
\end{proof}

\section*{Acknowledgments}

We thank the participants of the Dagstuhl Seminar 15122 `Formal Models
of Graph Transformation for Natural Language Processing' for
interesting discussions on the problem studied in this draft.

\bibliographystyle{plainnat}
\bibliography{peter}

\end{document}